\documentclass[letterpaper, 10 pt, conference]{ieeeconf}
\IEEEoverridecommandlockouts
\usepackage{amsmath}
\usepackage{amssymb}
\usepackage{graphicx}
\usepackage[colorlinks=true, allcolors=blue]{hyperref}
\usepackage[permil]{overpic}
\usepackage[compatibility=false]{caption}
\usepackage{subcaption}
\usepackage{comment}
\usepackage{xcolor}
\usepackage{algorithm}
\usepackage{algpseudocode}

\algnewcommand{\Initialize}[1]{%
  \State \textbf{Initialize:}
  \Statex \hspace*{\algorithmicindent}\parbox[t]{.8\linewidth}{\raggedright #1}
}
\algnewcommand{\Output}[1]{%
  \State \textbf{Output:}
  \Statex \hspace*{\algorithmicindent}\parbox[t]{.8\linewidth}{\raggedright #1}
}
\newtheorem{assumption}{Assumption}
\newtheorem{theorem}{Theorem}

\newtheorem{lemma}{Lemma}

\everymath{\displaystyle}
\graphicspath{ {Figures/} }

\title{\LARGE \bf Dissipative Avoidance Feedback for Reactive Navigation Under Second-Order Dynamics}
\author{Lyes Smaili$^{1}$ , Zhiqi Tang$^{2}$, Soulaimane Berkane$^{1,3}$, and Tarek Hamel$^{4}$
\thanks{This research work is supported in part by NSERC-DG RGPIN-2020-04759 and Fonds de recherche du Qu\'ebec (FRQ).}
\thanks{$^{1}$Department of Computer Science and Engineering, University of Quebec in Outaouais, Gatineau, QC, Canada. {\tt\small smal01@uqo.ca}, {\tt\small Soulaimane.Berkane@uqo.ca}}
\thanks{$^{2}$Division of Decision and Control Systems, School of Electrical Engineering and Computer Science, KTH Royal Institute of Technology, SE-10044 Stockholm, Sweden. {\tt\small ztang2@kth.se}}
\thanks{$^{3}$Department of Electrical Engineering, Lakehead University, ON, Canada.}
\thanks{$^{4}$I3S-UCA, CNRS, Universite C\^ote d’Azur \& Institut Universitaire de France (IUF), France. {\tt\small thamel@i3s.unice.fr}}
}

\date{February 2024}

\begin{document}

\maketitle
\thispagestyle{empty}
\pagestyle{empty}

\begin{abstract}
 This paper addresses the problem of autonomous robot navigation in unknown, obstacle-filled environments with second-order dynamics by proposing a Dissipative Avoidance Feedback (DAF). Compared to the Artificial Potential Field (APF), which primarily uses repulsive forces based on position, DAF employs a dissipative feedback mechanism that accounts for both position and velocity, contributing to smoother and more natural obstacle avoidance. The proposed continuously differentiable controller solves the motion-to-goal problem while guaranteeing collision-free navigation by using the robot's state and local obstacle distance information. We show that the controller guarantees safe navigation in generic 
$n$-dimensional environments and that all undesired $\omega$-limit points are unstable under certain \textit{controlled} curvature conditions. Designed for real-time implementation, DAF requires only locally measured data from limited-range sensors (e.g., LiDAR, depth cameras), making it particularly effective for robots navigating unknown workspaces. Simulations in 2D and 3D environments are conducted to validate the theoretical results and to showcase the effectiveness of our approach.
\end{abstract}

\section{Introduction}
Autonomous navigation with obstacle avoidance is a significant challenge in robotics, particularly in dynamic and unknown environments. Various techniques have emerged to address this problem, including the widely recognized artificial potential fields (APF) method \cite{khabib1986}. While APF is celebrated for its simplicity and efficiency compared to sampling-based methods like A* or Dijkstra, it is prone to local minima \cite{Koditschek1987}. The navigation function approach \cite{KODITSCHEK1990} offers a global solution by leveraging prior environmental knowledge and careful gain tuning to ensure almost global asymptotic stability (AGAS) in Euclidean spaces \cite{KoditschekRimon1992}. This method can be extended to complex environments using diffeomorphic mappings \cite{KoditschekRimon1992}, \cite{LiCailiTanner2019}, \cite{Paternain2017}, enhancing flexibility for non-trivial topologies. Other global methods, like the navigation transform \cite{Loizou2017} and prescribed performance control \cite{Vrohidis2018}, simplify parameter tuning while achieving safe navigation. Hybrid feedback approaches have also been developed to guarantee global asymptotic stability \cite{sanfelice2006robust, berkane2019hybrid, casau2019hybrid, berkane2021obstacle}, with recent strategies ensuring obstacle avoidance along the shortest path \cite{cheniouni2023safe}.

Navigation function-based methods have been adapted for unknown environments to facilitate real-time applications \cite{Lionis2007, Filippidis2011}. Reactive methods relying on local information, such as the sensor-based approach in \cite{arslan2019sensor}, achieve AGAS in environments with convex obstacles. Safety Velocity Cones (SVCs) have been proposed in \cite{berkane2021Navigation,smaili2024} as a simple and lightweight reactive feedback-based approach while more sophisticated hybrid feedback control methods in \cite{sawant2023hybrid} enable safe and global navigation in unknown environments.
Most of these methods, except for \cite{khabib1986}, target single-integrator systems, raising safety concerns when applied to robots with complex dynamics. An extension to higher-order models via reference governors is proposed in \cite{Arslan2017extension}, though it limits performance in less aggressive scenarios. Navigation functions for second-order systems are explored in \cite{Loizou2011}, ensuring safety but only guaranteeing stability in two-dimensional spaces. Additionally, distributed receding horizon planning for double-integrator systems \cite{Mendes2017uncertain} faces complexity challenges, and adaptive control methods \cite{VERGINIS2021adaptive} may be challenging to implement in unpredictable environments. Overall, achieving reliable and reactive obstacle avoidance under higher-order dynamics remains a challenge.

The objective of this work is to develop a generic control approach for navigation of autonomous robots under second-order dynamics in an $n$-dimensional environment filled with obstacles. The proposed controller combines a classical PD term with a dissipative term analogous to divergence flow, inspired by the behavior of flying insects, to decelerate the robot’s motion when approaching obstacles \cite{M.V.Srinivasan1996}. Note that the use of a dissipative term has been employed to address collision avoidance for multi-agent systems in the absence of obstacles \cite{Tang2023}, as well as  VTOL-UAV landing in obstacle-rich environments \cite{ROSA2014}. In contrast to these works, our approach addresses the motion-to-goal problem in the presence of multiple static obstacles, using only local information, and provides formal guarantees of collision avoidance and convergence under second-order dynamics. The proposed controller utilizes two simple pieces of information to avoid obstacles: the distance to the closest obstacle and the bearing to the obstacle. This makes the implementation of the controller inexpensive and versatile, enabling its application across various settings, including robots equipped with vision systems or LiDAR. We demonstrate that, under certain regularity conditions on the obstacles, such as smoothness and curvature constraints, the target equilibrium is asymptotically stable, while the undesired $\omega$-limit points are unstable. Overall, the proposed approach seeks to effectively address the challenges of reliable reactive navigation in unknown environments under second-order dynamics.

This paper is organized as follows: Section \ref{section:Notation} provides the notation used and different definitions used throughout this work. Section \ref{section:Problem Formualtion} describes the problem of interest, defines the workspace, and provides some assumptions on its topology. Section \ref{section:FeedbackControl} is devoted to the control design and stability analysis. Section \ref{section:NumericalSimulation} provides numerical simulations to show the effectiveness of the proposed obstacle avoidance approach in 2D and 3D environments. Finally, Section \ref{section:Conclusion} provides some concluding remarks.

\section{Notation and Preliminaries}\label{section:Notation}

Let $\mathbb{R}$, $\mathbb R_{>0}$, and $\mathbb N$ denote the set of reals, positive reals, and natural numbers, respectively. $\mathbb R^n$ represents the $n$-dimensional Euclidean space, with $n\in\mathbb N$. The Euclidean norm of a vector $p\in\mathbb R^n$ is represented as $\|p\|$. For a subset $\mathcal{A}\subset\mathbb R^n$, $\textbf{int}(\mathcal{A})$, $\partial\mathcal{A}$, $\overline{\mathcal{A}}$, and $\mathcal{A}^{\complement}$ denote its topological interior, boundary, closure, and complement in $\mathbb R^n$, respectively. Let $f(x):\mathbb{R}^n\to\mathbb{R}$ be a scalar-valued function and $F(x):\mathbb{R}^n\to\mathbb{R}^m$ be a vector field. $\nabla f(x)$ and $\nabla^2f(x)$ denote the gradient and Hessian of $f$ with respect to $x$, respectively. $\nabla F(x)$ represents the Jacobian of $F$ with respect to $x$. The Euclidean ball of radius $r>0$ centered at $p$ is defined as $\mathcal{B}(p,r):=\{y\in\mathbb R^n: \|p-y\|<r\}$.
Define the distance from a point $p\in\mathbb R^n$ to a closed set $\mathcal{A}\subset\mathbb R^n$ as
$d^0_\mathcal{A}(p):=\inf _{y\in\mathcal{A}}\|y - p\|$.
For two sets $\mathcal{A},\mathcal{B}\subset\mathbb{R}^n$, the distance between them is given by
$
    d^0_{\mathcal{A},\mathcal{B}}:= \inf _{p\in\mathcal{A},y\in\mathcal{B}}\|y - p\|.
$
 The projection of a point $p\in\mathbb{R}^n$ onto a set $\mathcal{A}\subset\mathbb R^n$ is denoted by 
 \begin{equation}
     \textbf{P}_\mathcal{A}(p):=\{y\in\overline{\mathcal{A}}: \|y-p\|=d^0_\mathcal{A}(p)\}.
 \end{equation}

We define the oriented distance function as $d_\mathcal{A}(p):= d^0_\mathcal{A}(p)-d^0_{\mathcal{A}^{\complement}}(p)$, see \cite[Definition 2.1, chap 7]{delfour2011shapes}. For the sake of simplicity, we will call $d_\mathcal{A}$ by the 'distance function' in the remaining text without referring to $d^0_\mathcal{A}$. The unit normal vector function, for all $p\in \mathbb R^n\setminus\{\partial\mathcal{A}\}$ such that $\textbf{P}_{\partial\mathcal{A}}(p)$ is unique, is given by \cite[Theorem 3.1, chap. 7]{delfour2011shapes}

\begin{equation}
    \nabla d_{\complement\mathcal{A}}(p):=\left\{
    \begin{array}{cc}
      -\frac{p-\textbf{P}_{\partial\mathcal{A}}(p)}{\|p-\textbf{P}_{\partial\mathcal{A}}(p)\|},   &  p\in\textbf{int}(\mathcal{A}^{\complement}),\\
      
      \frac{p-\textbf{P}_{\partial\mathcal{A}}(p)}{\|p-\textbf{P}_{\partial\mathcal{A}}(p)\|},   &  p\in\textbf{int}(\mathcal{A}).   
    \end{array}
    \right.
\end{equation}
The skeleton of $\mathcal{A}$ is defined as
\begin{equation}
    \mathbf{Sk}(\mathcal{A}):=\{p\in\mathbb R^n:\mathbf{card}(\mathbf{P}_{\overline{\mathcal{A}}}(p))>1\}.
\end{equation}
which is a set of points on $\mathbb R^n$ whose projection onto $\overline{\mathcal{A}}$ is not unique.
For a non-empty set $\mathcal{A}$, the reach of $\mathcal{A}$ at $p\in\overline{\mathcal{A}}$ is defined as
    \begin{align}
        &\mathbf{reach}(\mathcal{A},p):=\nonumber\\
        &\left\{\begin{array}{ll}
          0,   &  p\in\partial\overline{\mathcal{A}}\cap\overline{\mathbf{Sk}(\mathcal{A})},\\
          \sup\{r>0:\mathbf{Sk}(\mathcal{A})\cap\mathcal{B}(p,r)=\emptyset \},   & \text{otherwise.}
        \end{array}\right.
    \end{align}
    The reach of the set $\mathcal{A}$ is given by \cite[Def. 6.1, chap. 6]{delfour2011shapes}
    \begin{equation}
        \mathbf{reach}(\mathcal{A}):=\inf_{p\in\mathcal{A}}\{\mathbf{reach}(\mathcal{A},p)\}.
    \end{equation}
    The set $\mathcal{A}$ has {\it positive reach} if $\mathbf{reach}(\mathcal{A})>0$.  
We refer to \cite[Definition 3.1, Chap 2]{delfour2011shapes} for the definition of sets of class $\mathcal{C}^{k,l}$.
The relation between the class of the set $\mathcal{A}$ and the class of the distance function $d_\mathcal{A}$ is given in the following lemma; extracted from \cite[Theorem 8.2, Chap. 7] {delfour2011shapes}.
\begin{lemma}\label{lemma: classDistancefunction}
    For $k\ge2$ and $0\le l\le1$, if $\mathcal{A}$ is a set of class $\mathcal{C}^{k,l}$, then
    \begin{equation*}
        \forall p\in\partial\mathcal{A}, \exists\rho>0 \text{ such that } d_\mathcal{A}\in\mathcal{C}^{k,l}(\overline{\mathcal{B}(p,\rho)}).
    \end{equation*}
\end{lemma}

\section{Problem Formulation} \label{section:Problem Formualtion}

We consider a ball-shaped robot, with radius $R\in\mathbb R_{>0}$, and centered at $p\in\mathbb R^n$. Let $\mathcal{W}$ be a closed subset of the $n$-dimensional Euclidean space $\mathbb R^n$, which bounds the workplace. We consider $m$ smaller and compact subsets $\mathcal{O}_i\subset\mathcal{W}$, $i\in\{1,\cdots,m\}$, that represents the obstacles. We denote by $\mathcal{X}$ the free space given by
\begin{equation}
    \mathcal{X}:=\mathcal{W}\setminus\bigcup_{i=1}^{m}\textbf{int}(\mathcal{O}_i).
\end{equation}
The obstacle region is given by the complement of the free space $\mathcal{X}^{\complement}$. We define the outer obstacle $\mathcal{O}_0$ that surrounds the workspace $\mathcal{W}$ by
\begin{equation}
    \mathcal{O}_0:=\mathbb R^n\setminus\textbf{int}(\mathcal{W}).
\end{equation}
In many classical approaches to obstacle avoidance, the obstacles are usually \textit{parameterized by functions (anaytical expressions)}, where geometries are defined through a specific function or scalar variable, such as quadratic functions for ellipsoids, see \textit{e.g.,} \cite{Paternain2017}. However, as the number and complexity of the obstacles increase, the requirement to explicitly characterize each obstacle by a function further complicates the controller's design, especially if a sensor-based design is required. In contrast, the proposed approach leverages a distance function \textit{parameterized by the free space} itself. This alternative allows for treating arbitrary, potentially complex geometries, provided they are measurable. By using distance functions, we avoid the need to explicitly describe obstacles with parametric functions, simplifying the representation of free space. This not only reduces the computational burden but also offers more flexibility in handling environments with numerous and unknown shaped obstacles, making it a more scalable and adaptable solution for navigation tasks; this approach has been used \textit{e.g.,} \cite{khabib1986,berkane2021Navigation,smaili2024}.

Consider the distance function $d_{\mathcal{X}^{\complement}}$ to the obstacle set $\mathcal{X}^{\complement}$. The following is a feasibility constraint which ensures that the problem is formulated within defined limits:
\begin{assumption}[Positive Reach Set]\label{assumption:positive reach}
    Given the free space $\mathcal{X}$, there exist a positive real $h>0$ such that any point $p\in\mathcal{X}$, with $d_{\complement\mathcal{X}}(p)<h$, has a unique projection $\mathbf{P}_{\partial\mathcal{X}}(p)$.
\end{assumption}
Assumption \ref{assumption:positive reach} implies that, according to \cite[Theorem
6.3, Chap. 6]{delfour2011shapes}, the set $\mathcal{X}$ has a positive reach, {\it i.e.},
\begin{equation}
    \mathbf{reach}(\mathcal{X})>0.
\end{equation}


Moreover, the following assumption introduces a regularity condition on  $\mathcal{X}$, which is required in the proposed design.

\begin{assumption}[Smooth sets]\label{assumption:smoothBoundaries}
    Given the free space $\mathcal{X}$, there exist a positive real $\rho>0$ such that the functions $d_{\mathcal{X}^{\complement}}$, $\nabla d_{\mathcal{X}^{\complement}}$ and $\nabla^2 d_{\mathcal{X}^{\complement}}$ are all continuous for $p\in\mathbb R^n$ with $d_{\mathcal{X}^{\complement}}(p)<\rho$. 
\end{assumption}
According to Lemma \ref{lemma: classDistancefunction}, this assumption implies that the free space $\mathcal{X}$ is a set of class $\mathcal{C}^{2}$ with a bounded Hessian $\nabla^2d_{\mathcal{X}^{\complement}}$.

From here onward, and for the sake of simplicity, we define
\begin{align}
&d(p):=d_{\mathcal{X}^{\complement}}(p)-(R+\epsilon),\\
&\eta(p):=\nabla d_{\mathcal{X}^{\complement}}(p),\\
&\mathbf{H}(p):=\nabla^2 d_{\mathcal{X}^{\complement}}(p).
\end{align}

Furthermore, since the robot is of radius $R$, and for feasibility of the problem, the uniqueness of projection (guaranteed by Assumption \ref{assumption:positive reach}) and the local regularity of the distance function and its derivative (Assumption \ref{assumption:smoothBoundaries}) should be satisfied at least a distance $R$ away from the boundary of the obstacles. This is the purpose of the following assumption. 
\begin{assumption}[Feasibility]\label{assumption:feasibility}
    Let $h, \rho$ be positive reals that satisfy Assumptions \ref{assumption:positive reach} and \ref{assumption:smoothBoundaries} for the given free space $\mathcal{X}$. If $R$ is the radius of the robot,  then, $0<R<\min(h,\rho)$.
\end{assumption}

Let $\epsilon\in\mathbb R_{>0}$ be a design parameter representing a safety distance from the obstacle region. We define the \textit{practical free space} $\mathcal{X}_\epsilon$ as follows:

\begin{equation}\label{eq:practicalfreespace}
    \mathcal{X}_{\epsilon}:=\{p\in\mathbb{R}^n: d_{\mathcal{X}^{\complement}}(p)\ge R+\epsilon\}\subset\mathcal{X}.
\end{equation}

The practical set $\mathcal{X}_\epsilon$ is constructed by eroding the set $\mathcal{X}$ by a layer of size $R+\epsilon$. We impose the following condition on the safety margin $\epsilon$ to ensure feasibility, existence and uniqueness of $\eta(p)$ and $\mathbf{H}(p)$, while guaranteeing that the practical free space $\mathcal{X}_\epsilon$ is connected:
\begin{align}\label{condition:1}
    0<\epsilon<\min(h,\rho)-R.
\end{align}
Note that a small enough $\epsilon$ can be always found thanks to Assumption \ref{assumption:feasibility}.

Now, we consider the second-order robot dynamics
\begin{equation}\label{eq: SecondOrderSystem}
    \begin{cases}
    \dot{p}&=v,    \\
     \dot{v}&=u, 
    \end{cases} 
\end{equation}
where $p\in\mathbb R^n$ and $v\in\mathbb R^n$ are the position and the velocity of the robot, respectively, and $u\in\mathbb R^n$ is the control input (acceleration). 

The objective is to design a feedback control law $u=\kappa_\mathcal{X}(p,v):\textbf{int}(\mathcal{X}_\epsilon)\times\mathbb R^n\to\mathbb R^n$  depending only on the state $(p,v)$ measures (or estimates) and local information about the obstacles such that all trajectories of the closed-loop system,  starting in $\textbf{int}(\mathcal{X}_\epsilon)\times\mathbb R^n$, are safe (remains in $\mathcal{X}_\epsilon$), as shown in Figure \ref{fig:ObstacleAvoidance}. Moreover, we require asymptotic stability of the desired target $(p_d,0)$ and instability of any potential undesired equilibria. 
\begin{figure}[t]
    \centering
    \includegraphics[width=0.8\columnwidth]{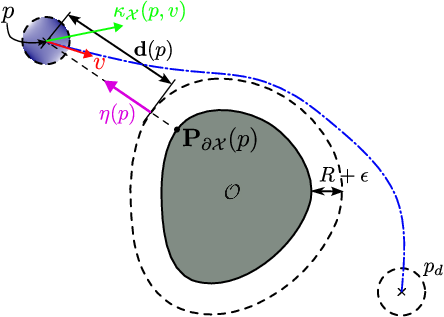}
    \caption{Illustration of the different variables and parameters used in the proposed control approach \eqref{eq:controller}.} 
    \label{fig:ObstacleAvoidance}
\end{figure}

\section{Dissipative Avoidance Feedback Control}\label{section:FeedbackControl}
In this section, we design a feedback control law that ensures safety and stability for a robot that navigates in environments with obstacles satisfying a given curvature condition which is presented later in this section. The proposed control law is derived from the Euler-Lagrange equation by incorporating the obstacle avoidance term in the dissipation function, compared to the artificial potential fields approach \cite{khabib1986} where the avoidance term is included in the potential energy, see Section \ref{section:NumericalSimulation}. 

First, we define the system's Lagrangian as 
$L(p,v):=T(v)-U(p)$, where $T(v):=v^\top v/2$ is the kinetic energy of the (unit mass) system and $U(p)=k_1\|p-p_d\|^2/2$ is the potential energy of the system, with $k_1$ a positive gain. We impose the following Rayleigh dissipation function 
\begin{equation}\label{eq:Ray}
    D(p,v):=D_{s}(v)+D_{a}(p,v),
\end{equation}
where the terms $D_{s}(v)$ and $D_{a}(p,v)$ correspond, respectively, to the stabilization and avoidance dissipation potentials, defined as follows
\begin{align}
D_{s}(v)&:=\frac{k_2}{2}v^\top v,\\
D_{a}(p,v)&:=\frac{k_3}{2}\gamma(d(p))v^\top \eta(p)\eta(p)^\top v \label{eq:Da},
\end{align}
for some positive gains $k_2$ and $k_3$ and some positive continuously differentiable scalar function $\gamma(\cdot):\mathbb R_{>0}\to\mathbb R_{>0}$ with $\gamma(z)\to+\infty$ when $z\to 0$.
The avoidance term $D_a(p,v)$ acts as a dissipative potential in the direction of the obstacle. Its effect vanishes when the velocity $v$ is orthogonal to the normal $\eta(p)$, ensuring that the robot's nominal motion is not altered. However, it blows up when $d(p)$ approaches zero, forcing the robot to decelerate when it gets closer to the obstacle (dissipation effect).


Plugging these energies into the Euler-Lagrange's equation:
\begin{equation}\label{eq:LagrangeEquation}
    \frac{d}{dt}\Big(\frac{\partial L}{\partial v}\Big)-\frac{\partial L}{\partial p}+\frac{\partial D}{\partial v}=0,
\end{equation}
results in the following second-order dynamics
\begin{equation}\label{eq:mvdot}
    \dot{v}=-k_1(p-p_d)-k_2v-k_3\gamma(d(p))\eta(p)\eta(p)^\top v.
\end{equation}
There are many possible choices for the function $\gamma(\cdot)$. In this work, the function $\gamma(.)$ is chosen as follows:
\begin{equation}\label{eq:smoothingFunction}
    \gamma(z):=\left\{\begin{array}{ll}
        0, & z\in(\epsilon_2,+\infty),\\
        \phi(z), & z\in[\epsilon_1,\epsilon_2),\\
        z^{-1}, &z\in(0,\epsilon_1),
    \end{array}
    \right.
\end{equation}

where $\epsilon_1, \epsilon_2$ are design parameters chosen such that
\begin{equation}\label{condition:3}
   {\color{blue} 0<\epsilon_1<\epsilon_2<\min(h,\rho)-(R+\epsilon),}
\end{equation}

and $\phi$ is a smooth enough function chosen such that $\gamma$ is a continuously differentiable function for all $z\in(0,+\infty)$\footnote{An example of the function $\phi$ is the cubic polynomial satisfying the conditions $\phi(\epsilon_1)=1/\epsilon_1$, $\phi'(\epsilon_1)=-1/\epsilon_1^2$, $\phi(\epsilon_2)=0$ and $\phi'(\epsilon_2)=0$, and given by
$
    \phi(z)=a+b(z-\epsilon_1)+c(z-\epsilon_1)^2+d(z-\epsilon_1)^3,
$
where
$   a=\epsilon_1^{-1}, b=-a^2, c=2a^2e-3ae^2, d=2ae^3-a^2e^2,e=(\epsilon_2-\epsilon_1)^{-1}.
$}.
To sum up, the proposed dissipative avoidance feedback (DAF) controller for \eqref{eq: SecondOrderSystem} is written as :
\begin{multline}\label{eq:controller}
    u=-k_1(p-p_d)-k_2v-k_3\gamma(d(p))\eta(p)\eta(p)^\top v.
\end{multline}
Note that the proposed controller is simple and can be implemented efficiently, as it only requires measurements of the minimum distance from the obstacles and the corresponding direction. The following theorem summarizes the contribution of this paper related to the safety and stability of the proposed controller.
\begin{theorem}\label{theorem:theorem1}
   Let $\epsilon$ be a design parameter associated to the practical free space $\mathcal{X}_\epsilon$ defined by \eqref{eq:practicalfreespace} and satisfies condition \eqref{condition:1}. Let $\epsilon_1$ and $\epsilon_2$ be the controller design parameters that satisfy condition \eqref{condition:3}. Consider the dynamics given by \eqref{eq: SecondOrderSystem} along with the controller $u$ given by (\ref{eq:controller}). Then, for any initial condition $(p(0),v(0))\in\textbf{int}(\mathcal{X}\epsilon)\times\mathbb R^n$ such that $D_{a}(p(0),v(0))$ in \eqref{eq:Da} is bounded, the following statements hold: \begin{enumerate}
        \item The set $\mathbf{int}(\mathcal{X}_\epsilon)\times\mathbb R^n$ is forward invariant.
        \item The robot state $(p,v)$ converges asymptotically to the desired equilibrium $(p_d,0)$ or reaches a set of isolated $\omega$-limit points $(p^*,0)$, where $p^*\in\mathcal{E}$, with: 
        \begin{equation}\label{eq:undiseredEquilibriumSet}
            \mathcal{E}:=\{p\in\partial\mathcal{X}_\epsilon:(p-p_d)=\mu\eta(p),\mu>0\}.
        \end{equation}
        \item For any $p^*\in\mathcal{E}$, if the Hessian $\mathbf{H}(p^*)$ admits an eigenvalue $\lambda_{\mathbf{H}}(p^*)$ such that 
        \begin{equation}\label{eq:conditionCompactForm}
            \lambda_{\mathbf{H}}(p^*)>\min\Big(1,\frac{k_2}{k_1}\Big)\frac{1}{\|p^*-p_d\|},
        \end{equation}
        then the $\omega$-limit point $(p^*,0)$ is unstable.
    \end{enumerate}
\end{theorem}
\begin{proof}
    See Appendix \ref{Appendix:C}.
\end{proof}
\begin{figure}[t]
    \centering
    \includegraphics[width=0.8\columnwidth]{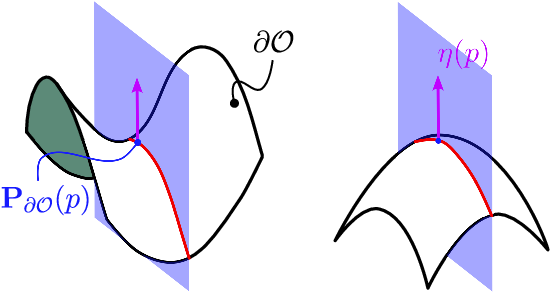}
    \caption{In 3D, the curvature condition \eqref{eq:conditionCompactForm} of Theorem \ref{theorem:theorem1} can be satisfied for convex or saddle-shaped obstacles provided that one principal direction satisfies the curvature condition \eqref{eq:conditionCompactForm} (\textit{e.g.,} strongly convex). } 
    \label{fig:CurvatureIn3D}
\end{figure}
\begin{figure}[t]
    \centering
    \includegraphics[width=0.85\columnwidth]{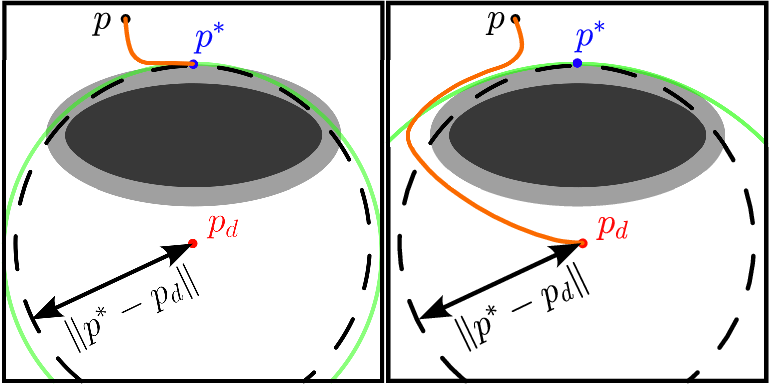}
    \caption{
    On the left, $k_2/k_1=0.93$ violates the curvature condition \eqref{eq:conditionCompactForm} making $p^*$ a local minimum. On the right, $k_2/k_1=0.68$ satisfies \eqref{eq:conditionCompactForm}, making $p^*$ a saddle point.}
    \label{fig:CurvatureVsGain}
\end{figure}
Theorem \ref{theorem:theorem1} states that safety is guaranteed under the proposed DAF controller \eqref{eq:controller}. Moreover, trajectories of the closed-loop system either converge to the desired equilibrium or one of the undesired $\omega$-limit points. The undesired $\omega$-limit points are shown to be unstable under the \textit{controlled} curvature condition \eqref{eq:conditionCompactForm}. This condition requires that one of the eigenvalues of the Hessian matrix $\mathbf{H}(p^*)$ be greater than some positive constant. Note that the Hessian matrix has at least one zero eigenvalue (since $\mathbf H(p^*)\eta(p^*)=0$) and the remaining eigenvalues correspond to the principal curvatures \footnote{The principal curvatures of an orientable hypersurface are the maximum and minimum values of the normal curvature (\cite[Definitions 2.2-2.3, Chap 5]{ShapeOperator}) as expressed by the eigenvalues of the shape operator at that point (\textit{i.e.,} the non-zero eigenvalues of $\mathbf{H}(p^*)$).} of $\partial\mathcal{O}_*$ at $p^*$. In 2D, and for $k_2/k_1>1$, condition \eqref{eq:conditionCompactForm} is equivalent to the {\it strong convexity} assumption used in the literature, \textit{e.g.,} \cite{arslan2019sensor,smaili2024}. The condition is controllable since the right-hand side of \eqref{eq:conditionCompactForm} can be made arbitrarily small by choosing $k_2/k_1$ small enough (see Figure \ref{fig:CurvatureVsGain}), although this might result in under-damped convergence. In 3D, only one principal direction for $\partial\mathcal{O}_*$ should have a positive principal curvature. This allows for avoiding saddle-shaped non-convex obstacles, as illustrated in Figure \ref{fig:CurvatureIn3D}. It is worth mentioning that the condition is always satisfied for spherical obstacles. It can also be regarded as a condition on the target position since there is always $p_d$ satisfying the condition for a fixed obstacle curvature.

\begin{figure*}
\centering
    \includegraphics[trim={3.1cm 0.1cm 2.8cm 0.8cm},clip,width=\textwidth]{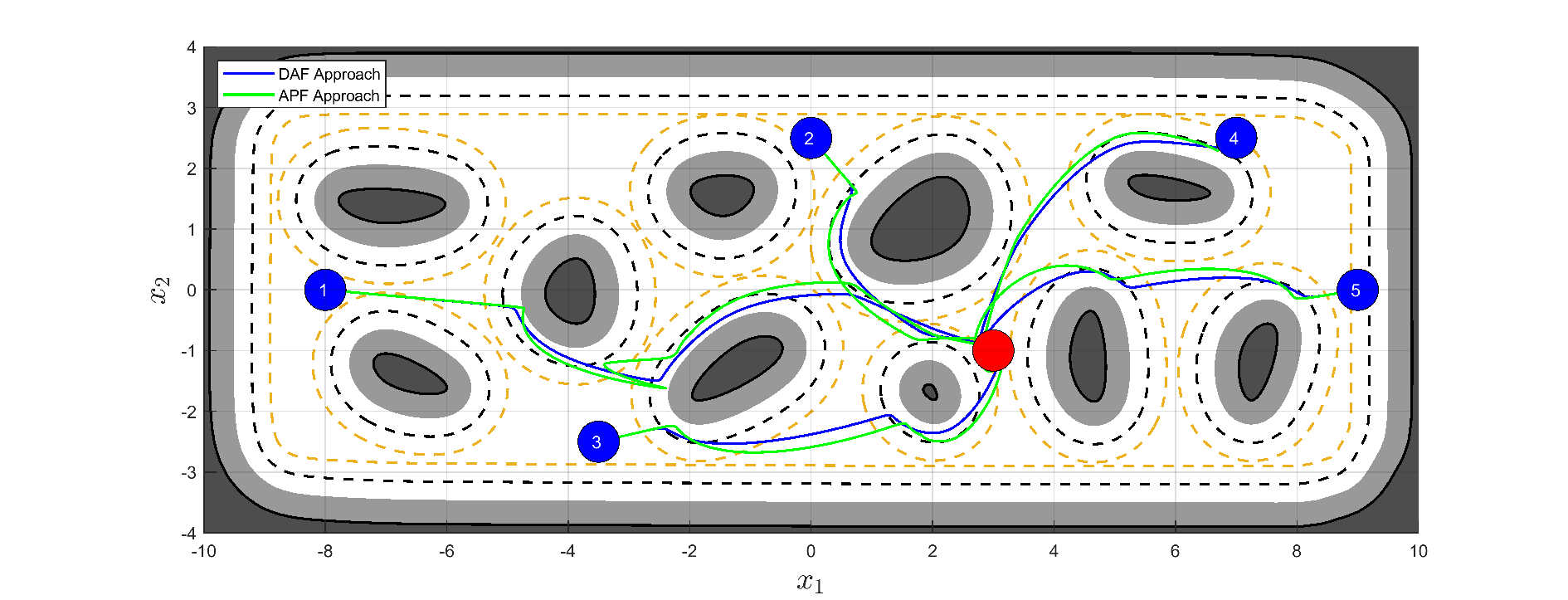}
    \caption{
    The trajectories of the robot in a 2D environment start from a set of initial positions (blue) away from the goal (red) while avoiding the obstacles (dark gray). The (light gray) regions are the dilation of the obstacles (dark gray) by the parameter $\epsilon$. The (black) and the (orange) dashed lines are a dilation by the parameters $\epsilon_1$ and $\epsilon_2$, respectively. The (blue) lines represent the trajectories under our approach, while the (green) lines represent the trajectories resulting from the APF approach. The animation is given at \href{https://youtu.be/7TfEdOQXw7Y}{https://youtu.be/7TfEdOQXw7Y}.
    } 
    \label{fig:2Dsimulatio}
\end{figure*}

The implementation of the proposed controller is summarized in Algorithm \ref{alg:DAF}, where some parameters tuning is taken into account for practical purposes, since the environment is supposed to be unknown. The dissipative term of the controller \eqref{eq:controller} requires to compute the normal $\eta(p)$ which must be unique and continuous. For these conditions to be satisfied, Assumptions \ref{assumption:positive reach} and \ref{assumption:smoothBoundaries} must hold for the given set $\mathcal{X}$. In a practical framework, the workspace and its parameters $\rho$ and $h$ are not supposed to be known and the initial choice of the parameters $\epsilon$, $\epsilon_1$ and $\epsilon_2$ does not necessarily satisfy  conditions \eqref{condition:1} and \eqref{condition:3}. We address this issue with parameters tuning as presented in Algorithm \ref{alg:DAF}, by checking the uniqueness of the projection $\mathbf{P}_{\partial\mathcal{X}}(p)$. If this projection is not unique, and as a result, the normal $\eta(p)$ is not unique, then we reduce the parameters $\epsilon$, $\epsilon_1$ and $\epsilon_2$ until the distance $d(p)\ge\epsilon_2$. At this point, the function $\gamma(.)$ reduces to zero and only the PD part of the controller remains which has no constraints to be computed. Furthermore, as long as the free space $\mathcal{X}$ fulfills Assumption \ref{assumption:feasibility}, then one can still find parameters $\epsilon$, $\epsilon_1$ and $\epsilon_2$ that will fit conditions \eqref{condition:1} and \eqref{condition:3}.

\begin{algorithm}
\caption{DAF Implementation}\label{alg:DAF}
\begin{algorithmic}[1]
\Initialize{$p(0)\in\mathcal{X}_\epsilon$, $p_d\in\mathbf{int}(\mathcal{X}_\epsilon)$, $v(0)$ \\ $\epsilon> 0$, $\epsilon_2 > \epsilon_1 > 0$}
\While{$\|p-p_d\|>\textrm{MaxError}$}
 \State Measure the minimal distance $d(p)$ using limited-range sensors ({\it e.g.,} Depth cameras, LiDAR, etc), and find the projections $\mathbf{P}_{\partial\mathcal{X}}(p)$.

    \If{$\mathbf{card}(\mathbf{P}_{\partial\mathcal{X}}(p))>1$} 
    \State Reduce $\epsilon> 0$,$\epsilon_1$ and $\epsilon_2$, with $\epsilon_2 > \epsilon_1 > 0$, until  $d(p)>\epsilon_2$. {\color{blue} \Comment{Trying to find parameters $\epsilon$, $\epsilon_1$ and $\epsilon_2$ that satisfy conditions \eqref{condition:1} and \eqref{condition:3}}}
    \If{Unable to reduce $\epsilon_2$ and $\epsilon_1$}
        \State Assumptions \ref{assumption:positive reach}-\ref{assumption:feasibility} are not satisfied for the current workspace.
    \EndIf
    \EndIf
    \If{$d(p)\leq\epsilon_2$} 
        \State Measure normal $\eta(p)$ corresponding to the minimal distance $d(p)$ to the obstacles set. {\color{blue} \Comment{The normal is only required when $d(p)\leq\epsilon_2$ to compute the controller}}
    \EndIf
    \If{$v \approx 0$ and $t \gg 1$}
        \State Reduce $k_2/k_1$. {\color{blue} \Comment{Condition \eqref{eq:conditionCompactForm} for the instability of $\omega$-limit points}}
    \EndIf
    \State Compute the DAF controller \eqref{eq:controller}.
\EndWhile
\end{algorithmic}
\end{algorithm} 

\section{Numerical Simulation}\label{section:NumericalSimulation}
\begin{figure*}
    \centering
    \includegraphics[trim={4cm 0.4cm 2cm 1.4cm},clip,width=\textwidth]{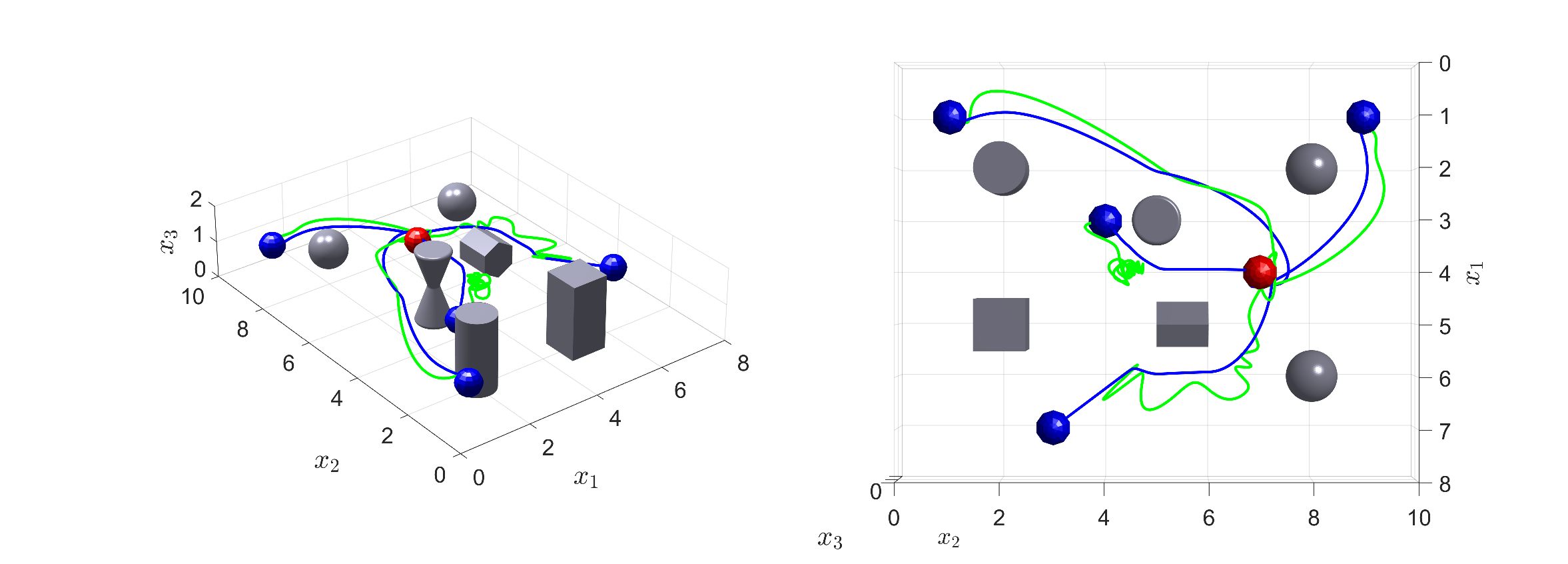}
    \caption{Trajectories of the robot, in a 3D environment filled with convex and non-convex obstacles, using our approach (blue) and the APF approach (green) starting at a set of initial positions (blue) away from the goal (red) while avoiding the obstacles (gray). Animated visualization can be found at \href{https://youtu.be/YqxY2tq8xTw}{https://youtu.be/YqxY2tq8xTw}. } 
    \label{fig:3Dsimulation}
\end{figure*}

In this section, we illustrate the performance of the proposed DAF controller in avoiding obstacles by performing 2D/3D numerical simulations. 
The performance of the approach is shown against the state-of-the-art Artificial potential Field (APF)  approach \cite{khabib1986}. The explicit APF controller is derived as follows. Consider the potential energy $U(p)=U_a(p)+U_r(p)$, where $U_a(p)$ and $U_r(p)$ are the attractive and repulsive, respectively, and defined as follows:
\begin{align}
    U_a(p)&:=\frac{k_a}{2}\|p-p_d\|^2,\\
    U_r(p)&:=\begin{cases}
        \frac{k_r}{2}\Big(\frac{1}{d(p)}-\frac{1}{\epsilon_2}\Big)^2,\hspace{7mm}d(p)\le\epsilon_2,\\
        0,\hspace{31mm}d(p)>\epsilon_2,
    \end{cases}
\end{align}
with $k_a$, $k_r$ positive gains. Letting the dissipation $D(v):=k_vv^\top v/2$ and evaluating the Euler-Lagrange's equation (\ref{eq:LagrangeEquation}), we end up with the following APF control law:
\begin{equation}
  u_{\textrm{APF}}=-k_a(p-p_d)-k_v v+k_r F_r(p),
\end{equation}
where the repulsive acceleration $F_r(p)$ is given by
\begin{multline*}
    F_r:=\begin{cases}
        \frac{k_r\eta(p)}{d^2(p)}\Big(\frac{1}{d(p)}-\frac{1}{\epsilon_2}\Big),\hspace{7mm}d(p)\le\epsilon_2,\\
        0,\hspace{36mm}d(p)>\epsilon_2.
    \end{cases}
\end{multline*}
For simulations in 2D, we consider the free space in Figure \ref{fig:2Dsimulatio} where the boundaries of the obstacles are defined by interpolating a few points with cubic splines. For simulations in 3D, the free space is illustrated in Figure \ref{fig:3Dsimulation} where smooth and non-smooth obstacles with different geometries are considered. For both simulations, the obstacles are assumed to be unknown. The relative distance and the normal to the obstacle are measured using a simulated 2D/3D LiDAR. 

\begin{figure}[t]
    \centering
    \begin{subfigure}{\columnwidth}
    \includegraphics[trim=0 0 0 0,clip,width=\columnwidth]{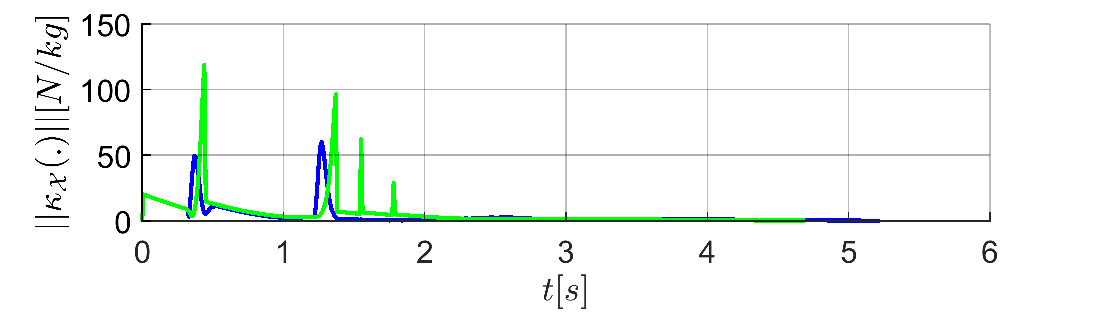}
    \caption{}
    \end{subfigure}
    
    \bigskip
    
    \begin{subfigure}{\columnwidth}
    \includegraphics[trim=0 0 0 0,clip,width=\columnwidth]{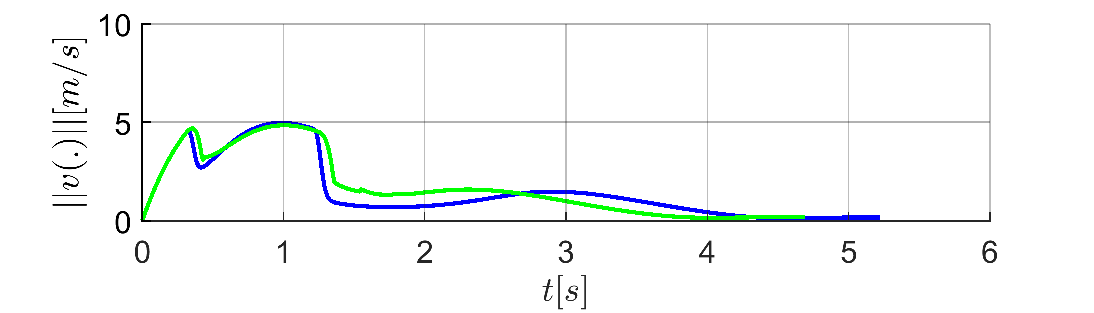}
    \caption{}
    \end{subfigure}
    \begin{subfigure}{\columnwidth}
    \includegraphics[trim=0 0 0 0,clip,width=\columnwidth]{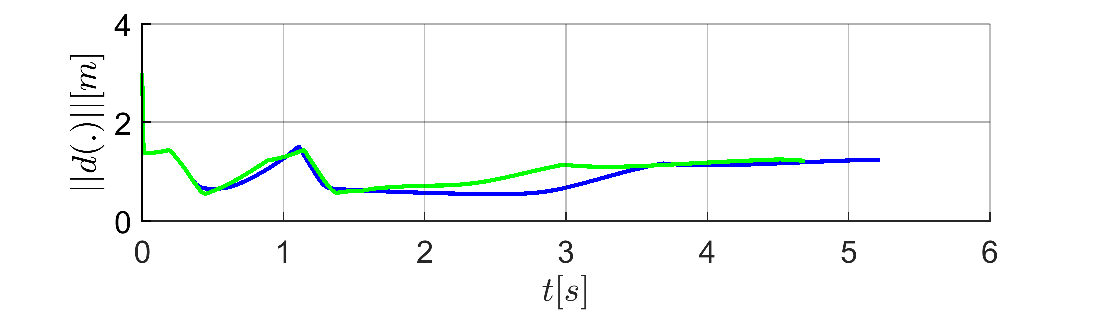}
    \caption{}
    \end{subfigure}
    \begin{subfigure}{\columnwidth}
    \includegraphics[trim=0 0 0 0,clip,width=\columnwidth]{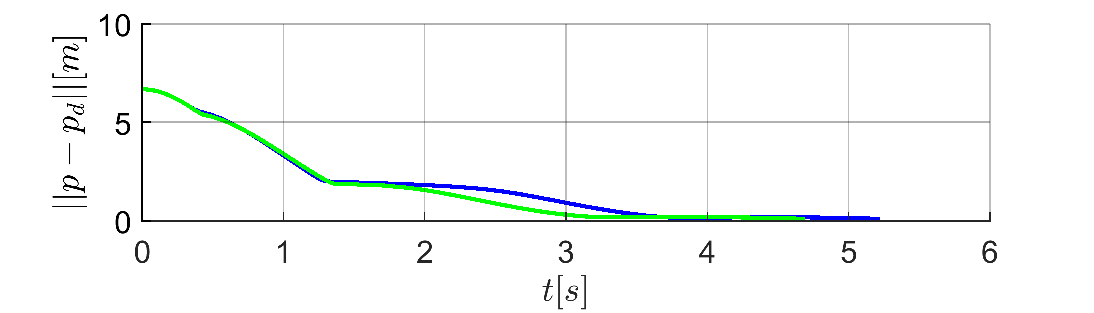}
    \caption{}
    \end{subfigure}
    \caption{(a) Evolution of the input norm  $\|u\|=\lVert \kappa_{\mathcal{X}}(.)\rVert$, (b) velocity norm $\lVert v(.)\rVert$,(c) distance to the obstacle $d(.)$ (d) and the error $\|p-p_d\|$ over time (case 3 presented in Figure \ref{fig:2Dsimulatio}). The (blue) plots correspond to the proposed DAF approach, and the (green) plots correspond to the APF method.}
    \label{fig:output}
\end{figure}

Figure \ref{fig:2Dsimulatio} shows the resulting trajectories of the robot in a 2D space, considering the desired equilibrium at $(3,-1)$, the robot's radius $R=0.34$, the controller parameters $\epsilon=0.06, \epsilon_1=0.3, \epsilon_2=0.6$, along with the gains  $k_1=3, k_2=2, k_3=8$, for the DAF approach, and $k_a=3, k_v=2, k_r=8$, for the APF approach. The maximum range of the 2D sensor is $R_s=3$. Note that we intentionally selected the gains so that the resulting velocity profiles of the DAF and APF approaches are nearly identical for the same initial conditions, as shown in Figure \ref{fig:output}. This choice allowed us to demonstrate that the robot behaves in a more natural and less aggressive manner (less acceleration required) under our approach. In contrast, the APF method often results in abrupt changes in motion direction due to the \textit{bumps} generated by the artificial potential field, requiring more aggressive maneuvers to avoid the obstacles using similar speeds. Finally, Figure \ref{fig:3Dsimulation} illustrates the resulting trajectories of the robot in a 3D space, considering the desired equilibrium at $(4,7,1)$, and the rest of the parameters are exactly the same as the 2D example. From this figure, it is fair to state that the trajectories resulting from the DAF approach seem more natural compared to the APF method. Since the gains are purposely chosen to be the same for the two approaches, the robot experiences a stronger repulsive force under the APF approach making it rebound near obstacles and resulting in oscillatory trajectories. In addition to that, even though the obstacles are separated enough, the robot can still get stuck in local minima under the APF controller (see the attached video) due to the same repulsive property. This reduces the chances of using this approach in crowded environments. On the other hand, the DAF method only requires the satisfaction of the condition \eqref{condition:3} for it to navigate the workspace freely. It is worth mentioning that, for the purpose of proving our theoretical results, the Assumption \ref{assumption:smoothBoundaries} was needed. This assumption restricts the geometry of the obstacles to having only smooth boundaries. This can be qualified as a sufficient condition, as the example presented in Figure \ref{fig:3Dsimulation} shows cases of non-smooth obstacles, yet our approach manages to work quite well.


\section{Conclusion}\label{section:Conclusion}
This paper addresses the challenge of safe navigation in unknown, obstacle-filled $n$-dimensional environments for robots operating under second-order dynamics. We proposed a Dissipative Avoidance Feedback (DAF) controller that guarantees the asymptotic stability of the target equilibrium and instability of undesired $\omega$-limit points. By utilizing only local information (distance and bearing), the proposed controller is well-suited for implementation on vehicles equipped with limited-range sensors. Compared to the state-of-the-art Artificial Potential Fields (APF) approach, DAF demonstrates a significant advantage by employing less aggressive maneuvers (lower acceleration). This opens the door to using the DAF approach in scenarios requiring high-speed navigation (\textit{e.g.} agile drone flights). Future work will explore complex robot dynamics, uncertainties, and adaptations to dynamic environments. 


\appendix

\subsection{Proof of Theorem \ref{theorem:theorem1}}
\label{Appendix:C}
For Item 1), one considers the following Lyapunov function candidate:
        \begin{equation}\label{eq:Lyapunov Function}
            {\cal L}(p,v)=\frac{k_1}{2}\|p-p_d\|^2+\frac{1}{2}\|v\|^2.
        \end{equation}
        One verifies that:
        \begin{equation}\label{eq:DerivativeLyapunov}
            \dot{{\cal L}}(p,v)=-k_2\|v\|^2-k_3\gamma(d(p))\dot{d}^2(p,v),
        \end{equation}
        is negative semi-definite as long as $d(p)>0$ ensuring that $p$ and $v$ are bounded. Recall that $\dot{d}(p,v)=\eta(p)^\top v$. Therefore, in view of the dynamical system (\ref{eq: SecondOrderSystem}) and the controller (\ref{eq:controller}), the dynamics of $\ddot{d}$ are given by: 
\begin{equation}\label{eq:distanceDynamics}
        \ddot{d}(p,v)=-k_3\Phi(d,\dot{d})-k_2\dot{d}-\alpha(p,v),
\end{equation}
with 
\begin{align}
\Phi( d,\dot{ d})&:=\gamma(d(p))\dot{d}(p,v),\\
\alpha(p,v) &:=k_1\eta(p)^\top(p-p_d)-v^\top\mathbf{H}(p)v. \label{eq:alpha}
\end{align}
Since the Hessian $\mathbf{H}(p)$, is continuous and bounded from Assumption \ref{assumption:smoothBoundaries}, it follows that $\alpha$ is bounded, $\forall t \in \mathbb R_{\geq 0}$. From there, direct application of \cite[Lemma 1]{Tang2023} shows that $\mathbf d(t) \in \mathbb R_{>0}$, $\forall t$. This shows that the set $\mathbf{int}(\mathcal{X}_\epsilon)\times\mathbb R^n$ is a forward invariant set. 

To prove Item 2), one first notes that the same Lemma also shows that $\Phi(\mathbf d,\dot{\mathbf d})$ and $\dot{\mathbf d}$ are bounded even if $\mathbf d(t) \in R_{>0}$ converges to zero.  This, in turn, implies that $\ddot{\cal L}$ is bounded, and hence, direct application of Barbalat's lemma ensures that $v \to 0$. To show that $p$ converges either to the desired point $p_d$
 or the undesired $\omega$-limit points $p^* \in \mathcal{E}$ defined by \eqref{eq:undiseredEquilibriumSet}, one distinguishes the case $d(t) \geq \bar{\epsilon} >0$ from the case for which $ d(t) \to 0$. For the first case, using Barbalat's lemma again, one verifies that $(p-p_d) \to 0$, and hence the equilibrium $(p,v)=(p_d,0)$ is asymptotically stable.
For the second case, the proof relies on a modified version of Barbalat's \cite[Lemma 1]{ms93rap}:
\begin{lemma}[\cite{ms93rap}]\label{lem:barbalat} 
Let $f(t)$ denote a solution to $\dot{f}(t)=g(t)+h(t)$ with $g(t)$ uniformly continuous function.
Assume that $\lim_{t \rightarrow \infty} f(t)=c$ and $\lim_{t \rightarrow \infty} h(t)=0$, with $c$ constant.
Then $\lim_{t \rightarrow \infty} \dot{f}(t)=0$.
\end{lemma}
Consider first the dynamics of $v$ in the directions orthogonal to $\eta$.  
Define $f_1(t)=\pi_\eta v$, where $\pi_x:=I_n-xx^\top$ for $\|x\|=1$ is the orthogonal projection operator, then one verifies:
\begin{align*}
\dot{f}_1=\dot\pi_\eta v +\pi_\eta \dot{v}&=-k_1 \pi_\eta (p-p_d) -k_2 \pi_\eta v + \dot\pi_\eta v,\\
&=g_1(t) + h_1(t)
\end{align*}
with $g_1(t)=-k_1 \pi_\eta (p-p_d) $ and $h_1(t)=-k_2 \pi_\eta v +\dot \pi_{\eta} v$. Using the fact $\dot{\eta}=\mathbf{H}(p)v$ which is bounded $\forall d(t) \in R_{\geq 0}$, one ensures that $g_1(t)$ is uniformly continuous and $h_1(t)$ is converging to zero. From there, one concludes that $\dot{f}_1$ also converges to zero. Hence $g_1(t)=-k_1 \pi_\eta (p-p_d) \to 0$, which implies that: 
\begin{equation}\label{eq:d=0} 
(p-p_d) \to \mu \eta, \quad \mbox{ with } \mu=\|p-p_d\|> 0.
\end{equation}
Consider now the dynamics of $v$ in the $\eta$ direction. The convergence of $d \to 0$ along with \eqref{eq:d=0}, implies that $p$ converges to $\mathcal{E}$ \eqref{eq:undiseredEquilibriumSet}, with $\mu=\|p^*-p_d\|$ and $\eta=(p^*-p_d)/\mu$. It also implies that $\alpha$  \eqref{eq:alpha} converges to a positive constant. From there, application of \cite[Lemma 1]{Tang2023} shows that $\ddot{d} \to 0$ and hence $\Phi(d,\dot{d})$  converges to $\alpha_\Phi=-\mu k_1/k_3$.

For Item 3), we consider the following positive function:
        \begin{equation}\label{eq:W}
           {\cal W}(p,v)=\frac{k_1}{2}||p-p^*||^2+\frac{1}{2}||v||^2.
        \end{equation}
        Recalling \eqref{eq:controller}, one verifies that:
        \begin{align}\label{eq:Wdot}
            \dot{\cal W}(p,v)=&-k_1 v^\top(p^*-p_d)\nonumber\\
            &-k_2||v||^2-k_3(v^\top\eta(p))\Phi(d(t),\dot{d}(t)).
        \end{align}

        Consider now an arbitrary point in the neighbor of $(p^*,0)$ given by
        \begin{equation}
            \left\{
            \begin{array}{l}
             v=  \sigma\delta v,\\
             p=p^*+\sigma\delta p,
            \end{array}\right.
        \end{equation}
        with $\delta p$ and $\delta v$ unit vectors tangent to $\partial\mathcal{X}_\epsilon$ and $\sigma$ a sufficiently small positive scalar. The gradient $\eta(p)$ can be approximated as follows:        
        \begin{equation}
            \eta(p)=\eta(p^*)+\sigma\mathbf{H}(p^*)\delta p+O(\sigma^2),
        \end{equation}
while $\Phi(d(p),\dot{d}(p,v))\approx \alpha_\Phi$.
With these ingredients and neglecting high order terms by choosing $\sigma$ small enough, one can verify that in the neighborhood of $(p^*,0)$:
        \begin{align}
            \dot{\cal W}(p,v)\approx&-k_2\sigma^2||\delta v||^2
            -\sigma^2 k_3\alpha_\Phi\delta v^\top\mathbf{H}(p^*)\delta p      
            \end{align}
By choosing $\delta v=\delta p=\nu_\lambda$, with $\nu_\lambda$ the eigenvector of $\mathbf{H}(p^*)$ associated to $\lambda_{\mathbf H}$,  one gets:
        \begin{align}
            \dot{\cal W}(p,v)&\approx-\sigma^2 \delta v^\top[k_2 I+\alpha_\Phi k_3\mathbf{H}(p^*)]\delta v, \notag \\
& \approx -\sigma^2 (k_2+\alpha_\Phi k_3\lambda_{\mathbf{H}})
        \end{align}
Now, to show that point $(p^*,0)$ is unstable, it suffices to show that $\dot{\cal W}(p,v)>0$ in the neighbor of that point and therefore:
        \begin{equation}\label{eq:condition 1}
            k_2+\alpha_\Phi k_3\lambda_{\mathbf{H}}<0 \Rightarrow \lambda_{\mathbf{H}}>\frac{k_2}{k_1\mu}=\frac{k_2}{k_1\|p^*-p_d\|}
        \end{equation}

The above condition also implies that there exists a point $p_0\in \partial\mathcal{X}_\epsilon $ in the neighbor of $p^*$ such that: ${\cal L}(p_0,0) \leq {\cal L}(p_0,0)$ and hence  $(\|p_0-p_d\|<\|p^*-p_d\|)$, where ${\cal L}(p_0,0)$ is the Lyapunov function defined by (\ref{eq:Lyapunov Function}). Geometrically speaking, there exists a plane defined by the three points $p_d$, $p_0$ and $p^*$ such that the intersection between $\partial\mathcal{X}$ around the obstacle and the plane is strictly inside the ball centered at $p_d$ with a radius $\|p^*-p_d\|$, except for $p^*$. This condition is also satisfied as long as there exists an eigenvalue $\lambda_{\mathbf{H}}$ of $\mathbf{H}(p^*)$ such that  
        \begin{equation}\label{eq:condition 2}
            \lambda_{\mathbf{H}}>\frac{1}{\|p^*-p_d\|}.
        \end{equation}
By combining (\ref{eq:condition 1}) and (\ref{eq:condition 2}) one concludes that:
        \begin{equation}
            \lambda_{\mathbf{H}}>\min\Big(1,\frac{k_2}{k_1}\Big)\frac{1}{\|p^*-p_d\|}.
        \end{equation}       

\bibliographystyle{ieeetr}
\bibliography{references}

\begin{thebibliography}{10}

\bibitem{khabib1986}
O.~Khatib, ``Real-time obstacle avoidance for manipulators and mobile robots,'' {\em The International Journal of Robotics Research}, vol.~5, no.~1, pp.~90--98, 1986.

\bibitem{Koditschek1987}
D.~Koditschek, ``Exact robot navigation by means of potential functions: Some topological considerations,'' in {\em Proceedings. IEEE International Conference on Robotics and Automation}, vol.~4, pp.~1--6, 1987.

\bibitem{KODITSCHEK1990}
D.~E. Koditschek and E.~Rimon, ``Robot navigation functions on manifolds with boundary,'' {\em Advances in Applied Mathematics}, vol.~11, no.~4, pp.~412--442, 1990.

\bibitem{KoditschekRimon1992}
E.~Rimon and D.~Koditschek, ``Exact robot navigation using artificial potential functions,'' {\em IEEE Transactions on Robotics and Automation}, vol.~8, no.~5, pp.~501--518, 1992.

\bibitem{LiCailiTanner2019}
C.~Li and H.~G. Tanner, ``Navigation functions with time-varying destination manifolds in star worlds,'' {\em IEEE Transactions on Robotics}, vol.~35, no.~1, pp.~35--48, 2019.

\bibitem{Paternain2017}
S.~Paternain, D.~E. Koditschek, and A.~Ribeiro, ``Navigation functions for convex potentials in a space with convex obstacles,'' {\em IEEE Transactions on Automatic Control}, vol.~63, no.~9, pp.~2944--2959, 2018.

\bibitem{Loizou2017}
S.~G. Loizou, ``The navigation transformation,'' {\em IEEE Transactions on Robotics}, vol.~33, no.~6, pp.~1516--1523, 2017.

\bibitem{Vrohidis2018}
C.~Vrohidis, P.~Vlantis, C.~P. Bechlioulis, and K.~J. Kyriakopoulos, ``Prescribed time scale robot navigation,'' {\em IEEE Robotics and Automation Letters}, vol.~3, no.~2, pp.~1191--1198, 2018.

\bibitem{sanfelice2006robust}
R.~G. Sanfelice, M.~J. Messina, S.~E. Tuna, and A.~R. Teel, ``Robust hybrid controllers for continuous-time systems with applications to obstacle avoidance and regulation to disconnected set of points,'' in {\em 2006 American Control Conference}, pp.~6--pp, IEEE, 2006.

\bibitem{berkane2019hybrid}
S.~Berkane, A.~Bisoffi, and D.~V. Dimarogonas, ``A hybrid controller for obstacle avoidance in an $ n $-dimensional euclidean space,'' in {\em the 18th European Control Conference (ECC)}, pp.~764--769, IEEE, 2019.

\bibitem{casau2019hybrid}
P.~Casau, R.~Cunha, R.~G. Sanfelice, and C.~Silvestre, ``Hybrid control for robust and global tracking on smooth manifolds,'' {\em IEEE Transactions on Automatic Control}, vol.~65, no.~5, pp.~1870--1885, 2019.

\bibitem{berkane2021obstacle}
S.~Berkane, A.~Bisoffi, and D.~V. Dimarogonas, ``Obstacle avoidance via hybrid feedback,'' {\em IEEE Transactions on Automatic Control}, vol.~67, no.~1, pp.~512--519, 2021.

\bibitem{cheniouni2023safe}
I.~Cheniouni, A.~Tayebi, and S.~Berkane, ``Safe and quasi-optimal autonomous navigation in sphere worlds,'' in {\em 2023 American Control Conference (ACC)}, pp.~2678--2683, IEEE, 2023.

\bibitem{Lionis2007}
G.~Lionis, X.~Papageorgiou, and K.~Kyriakopoulos, ``Locally computable navigation functions for sphere worlds,'' in {\em IEEE International Conference on Robotics and Automation, ICRA'07}, pp.~1998--2003, 2007.

\bibitem{Filippidis2011}
I.~Filippidis and K.~J. Kyriakopoulos, ``Adjustable navigation functions for unknown sphere worlds,'' in {\em the 50th IEEE Conference on Decision and Control and European Control Conference}, pp.~4276--4281, 2011.

\bibitem{arslan2019sensor}
O.~Arslan and D.~E. Koditschek, ``Sensor-based reactive navigation in unknown convex sphere worlds,'' {\em The International Journal of Robotics Research}, vol.~38, no.~2-3, pp.~196--223, 2019.

\bibitem{berkane2021Navigation}
S.~Berkane, ``Navigation in unknown environments using safety velocity cones,'' in {\em 2021 American Control Conference (ACC)}, pp.~2336--2341, 2021.

\bibitem{smaili2024}
L.~Smaili and S.~Berkane, ``Real-time sensor-based feedback control for obstacle avoidance in unknown environments,'' {\em arXiv preprint arXiv:2403.08614}, 2024.

\bibitem{sawant2023hybrid}
M.~Sawant, S.~Berkane, I.~Polushin, and A.~Tayebi, ``Hybrid feedback for autonomous navigation in planar environments with convex obstacles,'' {\em IEEE Transactions on Automatic Control}, 2023.

\bibitem{Arslan2017extension}
O.~Arslan and D.~E. Koditschek, ``Smooth extensions of feedback motion planners via reference governors,'' in {\em 2017 IEEE International Conference on Robotics and Automation (ICRA)}, pp.~4414--4421, 2017.

\bibitem{Loizou2011}
S.~G. Loizou, ``Closed form navigation functions based on harmonic potentials,'' in {\em 2011 50th IEEE Conference on Decision and Control and European Control Conference}, pp.~6361--6366, 2011.

\bibitem{Mendes2017uncertain}
J.~M. Mendes~Filho, E.~Lucet, and D.~Filliat, ``Real-time distributed receding horizon motion planning and control for mobile multi-robot dynamic systems,'' in {\em 2017 IEEE International Conference on Robotics and Automation (ICRA)}, pp.~657--663, 2017.

\bibitem{VERGINIS2021adaptive}
C.~K. Verginis and D.~V. Dimarogonas, ``Adaptive robot navigation with collision avoidance subject to 2nd-order uncertain dynamics,'' {\em Automatica}, vol.~123, p.~109303, 2021.

\bibitem{M.V.Srinivasan1996}
P.~S. Bhagavatula, C.~Claudianos, M.~R. Ibbotson, and M.~V. Srinivasan, ``Optic flow cues guide flight in birds,'' {\em Current Biology}, vol.~21, no.~21, pp.~1794--1799, 2011.

\bibitem{Tang2023}
Z.~Tang, R.~Cunha, T.~Hamel, and C.~Silvestre, ``Constructive barrier feedback for collision avoidance in leader-follower formation control,'' in {\em 2023 62nd IEEE Conference on Decision and Control (CDC)}, pp.~368--374, 2023.

\bibitem{ROSA2014}
L.~Rosa, T.~Hamel, R.~Mahony, and C.~Samson, ``Optical-flow based strategies for landing vtol uavs in cluttered environments,'' {\em IFAC Proceedings Volumes}, vol.~47, no.~3, pp.~3176--3183, 2014.
\newblock 19th IFAC World Congress.

\bibitem{delfour2011shapes}
M.~C. Delfour and J.-P. Zol{\'e}sio, {\em Shapes and geometries: metrics, analysis, differential calculus, and optimization}.
\newblock SIAM, 2011.

\bibitem{ShapeOperator}
B.~O'Neill, ``Chapter 5 - shape operators,'' in {\em Elementary Differential Geometry} (B.~O'Neill, ed.), pp.~202--262, Boston: Academic Press, second edition~ed., 2006.

\bibitem{ms93rap}
A.~Micaelli and C.~Samson, ``{Trajectory tracking for unicycle-type and two-steering-wheels mobile robots},'' Research Report RR-2097, {INRIA}, 1993.

\end{thebibliography}

\end{document}